\documentclass[a4paper,11pt]{article}

\usepackage[margin=1in]{geometry}
\usepackage[noblocks]{authblk}

\usepackage[utf8x]{inputenc}

\usepackage{amsmath,amssymb} 
\usepackage{verbatim}

\usepackage[colorlinks,linkcolor=blue,citecolor=green!50!black]{hyperref}

\usepackage{amsthm} % included in informs1.cls
\theoremstyle{plain}
	\newtheorem{theorem}{Theorem}[section]
	\newtheorem*{theorem*}{Theorem}
	\newtheorem*{lemma*}{Lemma}
  \newtheorem*{corollary*}{Corollary}
	\newtheorem{lemma}[theorem]{Lemma}
	
	\newtheorem{corollary}[theorem]{Corollary}
	
\theoremstyle{definition}
	\newtheorem{definition}[theorem]{Definition}
	
	\newtheorem{remark}{Remark}
	\newtheorem*{remark*}{Remark}
	\newtheorem{observation}[theorem]{Observation}
	
	\newtheorem*{observation*}{Observation}

\usepackage{tikz}
\usetikzlibrary{arrows}
\tikzset{>={angle 60}}

\usepackage{bbm}

 \newcommand{\ints}{\mathbb{Z}}
 \newcommand{\reals}{\mathbb{R}}
 \newcommand{\rats}{\mathbb{Q}}
\newcommand{\vecofones}{\mathbf {1}}

\renewcommand{\epsilon}{\varepsilon}
\newcommand{\eps}{\epsilon}
\renewcommand{\phi}{\varphi}
\renewcommand{\theta}{\vartheta}
\newcommand{\transp}{^{\text{\upshape{\tiny{T}}}}}
\newcommand{\floor}[1]{ \left\lfloor #1 \right\rfloor }
\newcommand{\ceil}[1]{ \left\lceil #1 \right\rceil }
\newcommand{\opt}{\text{\upshape{OPT}}}

\DeclareMathOperator*{\argmin}{argmin}

\newcommand{\calX}{\mathcal{X}}
\newcommand{\calY}{\mathcal{Y}}

\newcommand{\calP}{\mathcal{P}}
\newcommand{\calR}{\mathcal{R}}

\newcommand{\bigo}{\mathcal{O}}

\newcommand{\norm}[1]{\lVert #1 \rVert}
\newcommand{\Norm}[1]{\langle\!\langle #1 \rangle\!\rangle}
\newcommand{\CS}{\textup{CP}}
\newcommand{\RPM}{\textup{RP}}
\newcommand{\gap}{\textsc{Gap}}

\newcommand{\ycs}{y^*} %{y^{\textup{cs}}}
\newcommand{\yref}{y^*} %{y^{\textup{rps}}}
\newcommand{\ideal}{y^{\textup{id}}} %{y^*}
\newcommand{\refpt}{y^{\textup{rp}}} %{\widetilde{y}}
\newcommand{\rd}{r}

\title{Reference Point Methods and Approximation\\in Multicriteria Optimization}
\author{C. B\"using}
\affil{RWTH Aachen, \texttt{buesing@or.rwth-aachen.de}}
\author{K.S. Goetzmann}
\author{J. Matuschke}
\author{S. Stiller}
\affil{TU Berlin, \texttt{\{goetzmann,matuschke,stiller\}@math.tu-berlin.de}}

\date{}

\begin{document}

\maketitle

\begin{abstract}
Operations research applications often pose multicriteria problems. Mathematical research on
multicriteria problems predominantly revolves around the set of Pareto optimal solutions, while in
practice, methods that output a single solution are more widespread. In real-world multicriteria
optimization, reference point methods are widely used and successful examples of such methods. A
reference point solution is the solution closest to a given reference point in the objective space. 

We study the approximation of reference point solutions. In particular, we establish that
approximating reference point solutions is polynomially equivalent to approximating the Pareto set.
Complementing these results, we show for a number of general algorithmic techniques in single
criteria optimization how they can be lifted to reference point optimization. In particular, we lift
the link between dynamic programming and FPTAS, as well as oblivious LP-rounding techniques. The
latter applies, e.g., to \textsc{Set Cover} and several machine scheduling problems. 
\end{abstract}

% \begin{verbatim}
% structure:
% motivation 
% 	example Shortest Path
% 	identify a single, balanced solution
% 	considered since the 70s and used in MCDM tools, but not so many theoretical results
% 	(axiomatization by Voorneveld et al?)
% connection to Pareto set
% 	for polynomially encodable p all Pareto optimal solutions are CS
% notion of approximation
% 	why approximation w.r.t. the distance does not work
% 	introduce modified objective
% approximate Pareto set gives approximate CS
% 	(fairly short and easy proof)
% from approximate CS to approximate Pareto set
% 	why this is not so easy in general
% 	why this works for sub-ideal reference points
% example LP
% 	CS for cornered norm remains LP, also for sub-ideal ref points
% 	=> approximability of Pareto set follows immediately
% 	HOWEVER: GAP for LP is also trivial.
% example LP-based approximation
% 	an LP-based approximation algorithm with component-wise bounds on the rounded solution
% 	yields an approximation of CS with the same factor
% 	examples: Vertex Cover, Stochastic Scheduling (M"ohring et al 1997), ... (cf. Shmoys/Williamson)
% example WDHS
% 	2-approx via LP + enumeration
% example Shortest Path (and as a special case also WDHS)
% 	via approximate Pareto set
% 	or
% 	extension of Aissi (works for cornered norm only)
% 
% technical proofs in the end.
% \end{verbatim}

\section{Introduction}
%!TEX root=../mathprog-cs.tex

%In many typical application areas of combinatorial optimization, trade-offs between conflicting objectives %, such as time versus cost or pollution versus profit,
%play a crucial role. For example, consider a route guidance system assisting a car driver to find a route that minimizes both travel time and fuel consumption. 

In many applications of combinatorial optimization, trade-offs between conflicting objectives play a
crucial role. For example, route guidance systems are a classical application of the shortest path
problem. Yet, a good route guidance should allow the driver to make an informed choice to balance
travel time and fuel consumption.

It is well-known that even for this basic example, the bicriteria shortest path problem, the number
of Pareto optimal (i.e., non-dominated) solutions can grow exponentially with the size of the
network. Decision makers may have different preferences how much extra fuel to spend on less travel
time. Thus, a central task of multicriteria optimization is to either find a \emph{single} solution
based on a priori expressed trade-off preferences of the decision maker, or to identify a set of
solutions that is of manageable (in mathematical terms: polynomial) size but still reflects all
possible trade-off options at least approximately.

A straightforward way to a single solution is the weighted-sum method: The trade-off preferences are
specified by two non-negative weights for time and fuel consumption. The navigation system then
chooses a route minimizing the weighted sum of the two objectives. Unfortunately, this method
deprives the decision maker of essential solutions: Consider an instance with three possible routes
with corresponding objective value vectors $(10,1)$, $(6,6)$, and $(1,10)$, respectively. The route
with fuel consumption $6$ and travel time $6$ will never be the optimum for any choice of weights,
despite being a balanced and thus attractive alternative for many drivers.

Formally, this shortcoming of the weighted-sum approach means that it cannot reach every point of
the Pareto set. This motivates the concept of compromise solutions and reference point
solutions as defined by Yu~\cite{Yu1973}, which returns a solution closest to a given reference
point, where the 
distance is measured by some norm in the objective space. Compromise solutions use the component-wise optimum over 
all solutions as a
reference point. The trade-off preferences are reflected by the choice of the norm in the objective
space. Every point in the Pareto set is a reference point solution for some norm. Reference point
methods are widely used in practice, serving as a core concept of MCDM\footnote{Multicriteria
Decision Making} tools (cf.~Caballero et al.~\cite{Caballero2002} and Opricovic and
Tzeng~\cite{Opricovic2004} for particular examples and
Ehrgott et al.~\cite{Ehrgott2005} for an overview). Still, they did not attract a lot of theoretical
interest so
far.

%, and maintain the full spectrum of preference
%contained in a Pareto set while prompting a single solution for fixed
%preferences. 
We show that approximating reference point solutions is polynomially equivalent to approximating
the Pareto set as proposed by Papadimitriou and Yannakakis~\cite{Papadimitriou2000}. Further, we
provide general techniques for approximation algorithms, by means of which reference point solutions
can often be approximated with the same factor as the single-criterion problem, most notably for the
case of LP-rounding. A byproduct of our results are approximation algorithms for the Pareto sets of many hard
combinatorial optimization problems.

% This inability of the weighted-sum approach to reach well-balanced solutions motivates the concept of compromise solutions~\cite{Yu1973},  
% a technique designed for identifying a solution that represents a
% particularly good compromise taking into account the decision maker's
% preferences by minimizing the distance to the so-called ideal point
% w.r.t.\ a given norm. While compromise solutions and more general
% reference point methods are widely used in practice, serving as core
% concept of MCDM\footnote{Multicriteria Decision Making} tools
% (cf.~\cite{Caballero2002,Opricovic2004} for particular examples and
% \cite{Ehrgott2005} for an overview), the theoretical knowledge on
% compromise solutions is very limited. 

% In this paper, we address this issue by deriving theoretical results that link compromise solutions closely to  approximability of the Pareto set, in particular the approximation results by Papadimitriou and Yannakakis~\cite{Papadimitriou2000}, and by providing general techniques for designing approximation algorithms for compromise solutions.

%Before we state our results more formally, we introduce the necessary basic definitions and point to literature related to compromise solutions and approximability of Pareto sets.

\paragraph{Related work.}
Multicriteria optimization has a long tradition. The central notion of \emph{Pareto optimality} goes
back to works by Vilfredo Pareto in the late 19th and early 20th century. Ever since then, solution
concepts in multicriteria optimization have been studied. The notion of \emph{compromise solutions}
was introduced in 1973 by Yu~\cite{Yu1973} and further studied and extended in the following years
by Freimer and Yu~\cite{Freimer1976}, Gearhardt~\cite{Gearhart1979}, Choo and Steuer~\cite{Choo1983}
and many others. The concept was later extended to more general reference points and is incorporated
in many MCDM tools (cf.~Caballero et al.~\cite{Caballero2002}, Opricovic and
Tzeng~\cite{Opricovic2004}, Ehrgott et al.~\cite{Ehrgott2005}). Recently, Voorneveld et
al.~\cite{Voorneveld2011} gave an axiomatization of compromise solutions, in particular those
w.r.t.~the Euclidean norm. 

Also the approximation of Pareto sets has been studied for several decades now. It was
initiated by Hansen in 1979~\cite{Hansen1979}, followed by several publications on specific problems
such as shortest paths (Warburton~\cite{Warburton1987}) and scheduling (Cheng et
al.~\cite{Cheng1998}). More general results
on the existence and computability of approximate Pareto sets were presented by Safer in his PhD
thesis~\cite{Safer1992} in 1992, and in 2000 by Papadimitriou and
Yannakakis~\cite{Papadimitriou2000}. Some of our results are based on the latter. 

The results by Papadimitriou and Yannakakis~\cite{Papadimitriou2000} were extended by Vassilvitskii
and Yannakakis~\cite{Vassilvitskii2005} and, under stronger assumptions on the problems, further
improved by Diakonikolas and Yannakakis~\cite{Diakonikolas2007,Diakonikolas2008}.
The latter publication is particularly related to our results on the
equivalence between the approximability of the weighted sum problem and the approximability of the
Pareto set (Corollary~\ref{cor:single-approx-to-pareto-approx}), as the authors show a similar
statement for \emph{convex} approximate Pareto sets.

Also several other works have studied the relationship between approximate Pareto sets and aggregations of the objectives into one single objective, and are thus related to reference point methods. Ackermann et al.~\cite{Ackermann2007} use approximate Pareto sets to optimize an aggregation that is assumed to be (partially) differentiable. Their results are restricted to bi-objective problems, however.
Recently, Mittal and Schulz~\cite{Mittal2012b, Mittal2012} have used approximate Pareto sets to approximately optimize low-rank functions over polytopes and discrete sets. While one of our results can be seen as a special case of their framework, the remainder of our work also implies the reverse direction of their results: If one can approximately optimize a certain class of low-rank functions, one can also compute an approximate Pareto set.

Multicriteria optimization and in particular compromise solutions are also closely related to
robust optimization, in particular to min-max regret robustness. This connection has also been
noted and exploited by others, e.g.~Aissi et al.~\cite{Aissi2006,Aissi2007}. We extend some of
their results to reference point methods.

\paragraph{Our contribution.}
Our research mainly focuses on minimization problems, and we will restrict ourselves to this
setting throughout most parts of this paper. We note that this is \emph{not} without loss of
generality, and some of the results do not hold in the context of maximization. We discuss the
differences in Section \ref{sec:maximization}.

In Section \ref{sec:general}, we establish an algorithmic link between reference point solutions and
approximation of the Pareto set. As a main result, we show that approximating reference point
solutions, approximating compromise solutions, and approximating the Pareto set are polynomially
equivalent. An overview over the reductions that are proven in this paper is given in
Figure~\ref{fig:equivalence-graph}. We also show that any point in the Pareto set can be obtained as
reference point solution for two classes of popular norms with polynomially sized norm parameters,
extending a result by Gearhardt~\cite{Gearhart1979}.

Combining these results with an easy  constant factor approximation for reference points, through
optimization of the weighted sum, yields the following interesting corollary: For any discrete
minimization problem with a fixed number of linear criteria, there is a constant factor
approximation for the Pareto set if and only if there is a constant factor approximation for the
single-criterion version of the problem. The approximation guarantee of the thus obtained set is
increased by a factor of $k$ 	(the number of criteria), but it remains constant.
% Approximating the Pareto set through weighted sum increases the approximation factor by an
% additional factor of $k$, but it remains constant.

In Section \ref{sec:application}, we show how to solve the reference point
problem approximately for many combinatorial optimization problems. 
As a main result in this section, we show that single-objective approximations obtained by oblivious
LP-rounding directly can be transferred to approximation algorithms for reference point methods.
%Our main in this section result reads: every LP-rounding algorithm that gives an
% $\alpha$-approximation for the single-objective case (with some mild conditions on the rounding),
% can be turned into an $\alpha$-approximation for the reference point solution of the same
% combinatorial problem.
Along the way, we also prove that reference point solutions for linear objectives on convex sets can
be found efficiently. From this we get a short alternative proof to Papadimitriou and
Yannakakis~\cite{Papadimitriou2000} for the existence of an FPTAS for
the Pareto set of such problems.
Finally, we extend a technique by Aissi et al.~\cite{Aissi2006} from robust optimization to
multicriteria
optimization, allowing us to construct an FPTAS for reference point problems from pseudopolynomial
algorithms. %Finally, we show there is a 2-approximation for compromise solutions of weighted
% disjoint hitting set, a problem contained in many optimization problems that is trivial in the
% single-objective case but becomes hard in the multi-objective setting.

In Section \ref{sec:maximization} we analyze maximization problems and present both positive and
negative
answers to the question which of the results from Section \ref{sec:general} carry over to
maximization.

\section{Preliminaries}\label{sec:preliminaries}
%!TEX root=../mathprog-cs.tex

Throughout the paper, we let $\calP$ denote a multicriteria discrete optimization problem with $k$
objectives. As usual in multicriteria optimization, we assume the number of objectives to be
fixed. With the exception of Section~\ref{sec:maximization}, we consider only minimization
objectives. 
As we want to study approximation, we also restrict to non-negative objective values.
An instance $I$ of $\calP$ is thus given by the set of feasible solutions $\calX$ and the vector of 
objective functions $c: \calX \rightarrow \ints^k_{\geq0}$. The \emph{objective vector set}
of the instance is defined by $\calY := c(\calX) \subseteq \ints^k_{\geq0}$. A solution $y\in\calY$
is \emph{Pareto optimal} if there is no $y'\in\calY\setminus\{y\}$ with $y'\leq y$, where $y'\leq y$
is defined as $y_i\leq y'_i\;\forall\;i\in[k]$. By $[k]$ here and throughout the paper we denote the
set $\{1,2,\ldots,k\}$. The \emph{Pareto set} $\calY_P$ is the set of all Pareto optimal solutions.

Similar to Papadimitriou and Yannakakis~\cite{Papadimitriou2000}, we will assume throughout this
paper that for any instance $I$,
we can compute an exponential bound on the objective values of all solutions, i.e., a number $M>0$
such that $\calY \subseteq [0,M]^k$ and such that there is a polynomial $\pi$ with
$M\leq2^{\pi(|I|)}$, where $|I|$ is the encoding length of the instance. This is not a major
restriction in usual discrete optimization problems.

\paragraph{Reference point methods.}
To model the decision maker's preferences, reference point methods take two types of additional input: a \emph{reference point} $\refpt \in
\ints_{\geq0}^k$ and a \emph{weight vector} $\lambda \in \rats^{k}_{\geq0}$ on the objectives.
The reference point  is a---usually unattainable---vector of aspired values for each criterion.
The weights are used to adjust a fixed norm $\norm{\cdot}$ on $\reals^k$ by letting
$\norm{\cdot}^\lambda$ be the norm defined by $\norm{y}^\lambda := \norm{(\lambda_1 y_1, \ldots,
\lambda_k y_k)}$.

The goal is to find a solution that is as close as possible to the reference point w.r.t.~$\norm{\cdot}^\lambda$. Conceive of this distance as the
price to pay to attain a compromise among the criteria. The objective value of an optimal reference
point solution is \emph{the value of the reference point, degraded by the price of compromise}.
For minimization, the reference point objective function thus reads:
$$\rd_{\refpt, \lambda}(y) = \norm{\refpt}^\lambda + \norm{y-\refpt}^\lambda\;.$$
Of particular interest in this context is the \emph{ideal point} $\ideal \in \ints_{\geq0}^k$,
which is defined as the point in the objective space obtained by optimizing each objective
individually, i.e., $\ideal_i := \min_{y\in\calY} y_i$. Throughout this paper, we will restrict ourselves to reference points
$\refpt$ with $\refpt \leq \ideal$. We call these points \emph{feasible reference points}.

%Besides choosing the reference point, the decision maker furthermore
%expresses her preferences by specifying weights. For any norm $\norm{\cdot}$ on $\reals^k$ and
%$\lambda \in \rats_{\geq0}^k$, let $\norm{\cdot}^\lambda$ be the norm defined by $\norm{y}^\lambda
%:= \norm{(\lambda_1 y_1, \ldots, \lambda_k y_k)}$. The reference point objective function therefore
%reads $\rd_{\refpt, \lambda}(y) = \norm{\refpt}^\lambda + \norm{y-\refpt}^\lambda$. 
%In case, the norm is immaterial or clear from the context we shall omit the superscript $\lambda$.

Formally, we define the problem of reference point solutions, $\RPM(\calP, \norm{\cdot})$ for short, as follows:
Given an instance of $\calP$, a feasible reference point $\refpt \in \ints^k_{\geq 0}$, and a weight
vector $\lambda \in \rats^k_{\geq0}$ as input, find a solution $x \in \calX$ that minimizes
$\rd_{\refpt, \lambda}(c(x))$. Given the particular interest of the ideal point, we will also consider the problem $\CS(\calP,
\norm{\cdot})$, which is known as \emph{compromise programming}: Given an instance of $\calP$
and $\lambda \in \rats^k_{\geq0}$, find a solution $x \in \calX$ that minimizes $\rd_{\ideal,
\lambda}(c(x))$.

\paragraph{The constant in the objective.}

As $\norm{\refpt}$ is a constant, exact minimization of $\rd(y)$ boils down to minimizing the
distance $\norm{y-\refpt}$, as the level sets of this function are identical to that of the
reference point objective function. Still, for judging the quality of an approximation, this
short-cut is not permissible, as the following
trivial example shows. Consider a multicriteria problem defined by $k$ unrelated copies of a single
criteria optimization problem, for which we have a tight approximation algorithm with factor
$\alpha$. Let the distance be measured in any norm, and choose the ideal point as a
reference point. As the single criteria problems are unrelated, one expects that solving each
problem separately by the approximation algorithm gives an $O(\alpha)$-approximation for the
reference point solution. This is indeed true for the reference point objective function.
However, for minimizing the distance, the ratio to the optimum is infinite, because the optimum
attains the ideal point for the unrelated problems.  

Conversely, any approximation algorithm for the
distance $\norm{y-\refpt}$ could be turned into an algorithm that solves the single-criterion
problem exactly (as the minimal distance to the ideal point is $0$ when focusing on a single
criterion). Thus, we can not hope for approximating the distance $\norm{y-\refpt}$ for any
problem that is NP-hard in the single criterion version. In contrast to that, for the objective
$\rd(y)$ we do get positive approximation results also for NP-hard problems.

\paragraph{Caveat on complexity.}
Note that although the concept of reference point solutions is a generalization of compromise
solutions, in terms of complexity \CS\ is \emph{not} a special case of \RPM. In the former problem,
the ideal point is not given, while in the latter case the reference point is given in the input.
This leads to different consequences if the underlying single-criterion problem can not be solved in
polynomial time. In this case, the objective function of {\CS} is hard to evaluate. However, in the
context of approximability this is only a minor issue, as
Corollary~\ref{cor:pareto-gives-compromise} shows.
For \RPM, on the other hand, it becomes hard to verify feasibility of the input (i.e., checking
whether $\refpt \leq \ideal$). The best we can expect from an algorithm is to \emph{approximately
distinguish} between feasible and infeasible instances, i.e., an $\alpha$-approximation algorithm
needs to accept all feasible inputs and reject all instances where $\refpt_i > \alpha \cdot
\ideal_i$ for some $i \in [k]$, but it might also accept instances with slightly infeasible
reference points, as long as $\refpt \leq \alpha \ideal$.

\paragraph{Norms.}
Throughout this paper, we will restrict  to norms fulfilling the following two properties. A norm
$\norm{\cdot}$ is called \emph{monotone}, if $y' \leq y''$ implies $\norm{y'} \leq \norm{y''}$ for
any $y', y'' \in \reals^k_{\geq0}$. It is called \emph{polynomially decidable}, if we can decide
whether $\norm{y'} \leq \norm{y''}$ in time polynomial in the encoding length of $y'$ and $y''$. 

We will mainly use the following families of norms: the infinity-norm
$\norm{y}_\infty := \max_i |y_i|$ (which we will sometimes also denote by $\Norm{y}_\infty$ for
convenience), the standard $\ell^p$-norm $\norm{y}_p := (\sum_i |y_i|^p)^\frac{1}{p}$, and the
cornered $p$-norm $\Norm{y}_p := \max_i |y_i| + \frac{1}{p}\sum_i |y_i|$ (both for $p\geq1$). The
cornered norm has been considered in the context of compromise programming before, e.g.
by Gearhart~\cite{Gearhart1979}. Our motivation to use this norm is twofold. Firstly, for
general values of $p$, it will be hard to minimize a distance measured in the $\ell^p$-norm because
of the exponents. The cornered $p$-norms are simpler, but still have properties similar to the
$\ell^p$-norms: Their unit spheres are nested within each other, and for increasing values of $p$
they approach the axis parallel square. This allows to control the degree of balancing of the
criteria in the reference point solution. Secondly, the infinity-norm (often referred to as
Chebyshev-norm in this context) is very popular in MCDM-tools. Often it is \emph{augmented} by a
small linear term to avoid weakly Pareto optimal solutions (cf.~Choo and Steuer~\cite{Choo1983}),
similar to the
addition of the term $\frac{1}{p}\norm{y}_p$.

Note that all $\ell^p$- and cornered $p$-norms are monotone and polynomially decidable.

\paragraph{Approximation of the Pareto set.}
We extend the well-known concept of approximation algorithms for the single-objective case to
approximability of the Pareto set in a similar way as done in
Papadimitriou and Yannakakis~\cite{Papadimitriou2000}, with the
slight difference of including constant factor approximations. For $\alpha>1$, an
\emph{$\alpha$-approximate Pareto set} is a set $\calY_\alpha\subseteq\calY$ such that for all
$y\in\calY_P$ there is $y'\in\calY_\alpha$ with $y'\leq\alpha y$. An \emph{$\alpha$-approximation
algorithm for the Pareto set} is an algorithm that constructs an $\alpha$-approximate Pareto set in
time polynomial in the encoding length of the instance of $\calP$. An \emph{FPTAS for the Pareto
set} is a family of algorithms that, for all $\eps>0$, contains a $(1+\eps)$-approximation algorithm
for the Pareto set with running time polynomial in  $\frac{1}{\eps}$ and the encoding length of the
instance.

\section{Equivalence of Approximation}\label{sec:general}
%!TEX root=../mathprog-cs.tex

In this section, we investigate the relation between approximation of the Pareto set, reference point methods, and compromise programming. Our main theorem states these three notions of approximability are essentially equivalent: A constant approximation factor for one of these problems implies constant (although possibly different) approximation factors for the others, and the same is true for approximation schemes.

\begin{theorem}\label{thm:approximation-equiv}
Let $\calP$ be a multicriteria discrete minimization problem. The following statements are
equivalent.
\begin{itemize}
\item There is a constant factor approximation (FPTAS, respectively) for the Pareto set of $\calP$.
\item There is a constant factor approximation (FPTAS, respectively) for $\RPM(\calP, \norm{\cdot})$
for every monotone and polynomially decidable norm $\norm{\cdot}$.
\item There is a constant factor approximation (FPTAS, respectively) for $\RPM(\calP, \norm{\cdot}_\infty)$.
\item There is a family of algorithms that, for each $p \geq 1$, contains a constant factor approximation (FPTAS, respectively) for $\RPM(\calP, \norm{\cdot}_p)$ or $\RPM(\calP, \Norm{\cdot}_p)$, and the running time of all algorithms is bounded by a polynomial in the input size and $\log(p)$.
\item There is a constant factor approximation (FPTAS, respectively) for $\CS(\calP, \norm{\cdot})$
for every monotone and polynomially decidable norm $\norm{\cdot}$.
\item There is a constant factor approximation (FPTAS, respectively) for $\CS(\calP, \norm{\cdot}_\infty)$.
\item There is a family of algorithms that, for each $p \geq 1$, contains a constant factor approximation (FPTAS, respectively) for $\CS(\calP, \norm{\cdot}_p)$ or $\CS(\calP, \Norm{\cdot}_p)$, and the running time of all algorithms is bounded by a polynomial in the input size and $\log(p)$.
\end{itemize}
\end{theorem}
Figure~\ref{fig:equivalence-graph} explicitly depicts the reductions we prove in the remainder of
this paper.

\begin{figure}[ht]
	\begin{center}
% 	\FIGURE
		\begin{tikzpicture}
			[xscale=2.5,yscale=2.5,textnode/.style={rounded corners,draw=black,inner sep=5pt,align=left}]
			\footnotesize			
			\node[textnode](pareto) at (0,1) {Pareto set};
			\node[textnode](gap) at (2,1) {\gap\ problem};
			\node[textnode](ws) at (4,1) {$\min_y \lambda\transp y$\\(weighted sum)};
			\node[textnode](rpm-gen) at (0,2) {$\RPM(\norm{\cdot})$,\\ $\norm{\cdot}$ monotone\\\&
poly decidable};
			\node[textnode](rpm-p) at (2,2) {$\RPM(\norm{\cdot}_p)$ and\\$\RPM(\Norm{\cdot}_p)$,
$p\geq1$};
			\node[textnode](rpm-inf) at (4,2) {$\RPM(\norm{\cdot}_\infty)$};
			\node[textnode](cs-gen) at (0,0) {$\CS(\norm{\cdot})$,\\ $\norm{\cdot}$ monotone\\\&
poly decidable};
			\node[textnode](cs-p) at (2,0) {$\CS(\norm{\cdot}_p)$ and\\ $\CS(\Norm{\cdot}_p)$,
$p\geq1$};
			\node[textnode](cs-inf) at (4,0) {$\CS(\norm{\cdot}_\infty)$};
			\draw[<->] (pareto) -- node[above]{\cite{Papadimitriou2000}} (gap);	
			\draw[->] (gap) -- node[above]{binary search} (ws); 
			\draw[->] (cs-gen) -- node[above]{} (cs-p);
			\draw[->] (cs-p) -- node[above]{} (cs-inf);
			\draw[->] (pareto) -- node[right]{Cor.~\ref{cor:pareto-gives-compromise}} (cs-gen);
			\draw[->] (cs-p) -- node[left]{Cor.~\ref{cor:cs_p-to-pareto}} (gap);
			\draw[->] (cs-inf) -- node[sloped, above]{Cor.~\ref{cor:cs-to-pareto}} (gap);
			\draw[dashed,->] (ws) -- node[right]{Thm.~\ref{thm:weighted-sum-approx}} (cs-inf);
%node[right]{additional factor $k$}
			\draw[->] (rpm-gen) -- node[below]{} (rpm-p);
			\draw[->] (rpm-p) -- node[below]{} (rpm-inf);
			\draw[->] (pareto) -- node[right]{Cor.~\ref{cor:pareto-rpm}} (rpm-gen);
			\draw[->] (rpm-p) -- node[left]{Lemma \ref{lem:ref_p-to-pareto}} (gap);
			\draw[->] (rpm-inf) -- node[sloped, above]{Lemma \ref{lem:ref-to-pareto}} (gap);
			\draw[dashed,->] (ws) -- node[right]{Thm.~\ref{thm:weighted-sum-approx}} (rpm-inf);
%node[right]{additional factor $k$}
		\end{tikzpicture}
		\caption{A graph of the reductions of approximability. \label{fig:equivalence-graph}
	An arrow from node A to node B indicates
that whenever there is a constant factor approximation algorithm for A, there also is a
constant factor approximation algorithm for B. With exception of the dashed arrows the implication
also holds for
approximation schemes. The result on the implication from weighted sum to $\CS$ and $\RPM$ is given
in Theorem~\ref{thm:weighted-sum-approx}.}
	\end{center}
\end{figure}
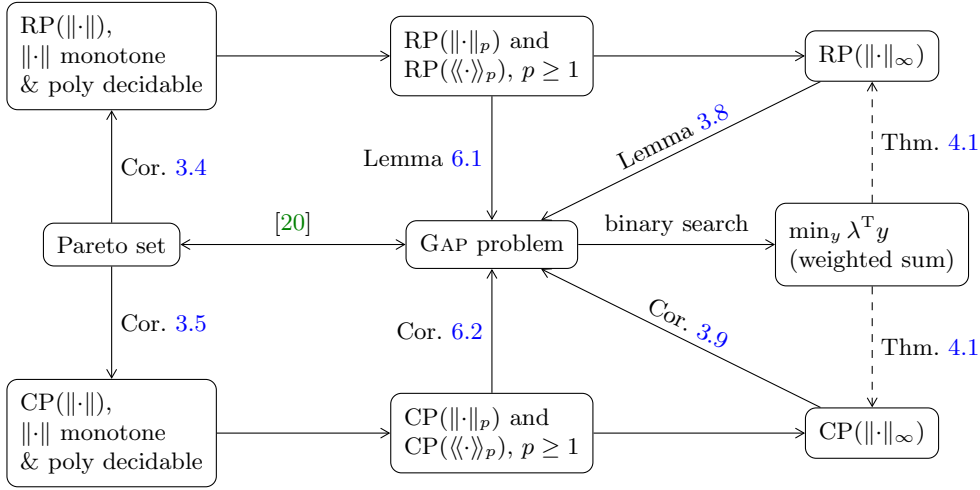

%We split the proof of this theorem in two parts. First, we show how to obtain an approximation for $\RPM$ (or $\CS$) from an approximate Pareto set. We then show how to use an approximation algorithm for $\RPM$ (or $\CS$) to obtain an approximation algorithm for the Pareto set.
Before we discuss these reductions, that together prove Theorem~\ref{thm:approximation-equiv}, in
detail, we turn our attention to a result of independent interest, that motivates the algorithmic
use of the $\ell^p$- and cornered $p$-norms.

\paragraph{Reference point solutions and the Pareto set.} 
%\fxnote{state that there always is a Pareto optimal CS}
Gearhardt~\cite{Gearhart1979} showed that for both the $\ell^p$-norm and the
cornered $p$-norm, if $p$ tends to infinity, the distance between the Pareto set and the set of
compromise solutions with respect to all non-negative normalized weight vectors tends to zero. This
means that for discrete optimization problems with a finite set of feasible solutions, there is a
finite value $p_0$ for which the two sets coincide. We show that, under the assumption that there is
an exponential bound on the objectives, also $p$ can be chosen
in such a way that it is polynomially encodable.

\begin{theorem}\label{thm:cs-gives-all-pareto}
    Let $\refpt \in \ints_{\geq0}^k$ be a feasible reference point. If the objective vector set
$\calY$ is contained in $[0,M]^k$, then the following statements hold true.
	\begin{enumerate}
		\item If $p>\frac{\log k}{\log(1+\frac{1}{M})}$, then for any Pareto optimal solution
$y\in\calY$ there is a weight vector $\lambda\in\rats_{\geq0}^k$ such that $y$ minimizes $\norm{y -
\refpt}^\lambda_p$.
		\item If $p>kM$, then for any Pareto optimal solution $y\in\calY$ there is a weight vector
$\lambda\in\rats_{\geq0}^k$ such that $y$ minimizes $\Norm{y - \refpt}^\lambda_p$.
	\end{enumerate}
\end{theorem}
\begin{proof}%[of \Cref{thm:cs-gives-all-pareto}]
	We first consider the cornered norm $\Norm{y}_p^\lambda = \max_{i\in[k]}\{\lambda_i y_i\} +
\frac{1}{p} \sum_{i\in[k]} \lambda_i y_i$. Therefore, let $p>kM$. Further let
$y\in\calY$ be a Pareto optimal cost vector, and let $I:=\{i\in[k]: y_i = \refpt_i\}$. We set the
weight vector $\lambda$ as follows:
\[ \lambda_i = \begin{cases} 1+k &\text{if } i\in I, \\ \frac{1}{y_i-\refpt_i} &\text{otherwise.}
\end{cases} \]
The weighted distance of $y$ to the reference point is
\[ \Norm{y-\refpt}_p^\lambda = \max_{i\notin I}\{\lambda_i(y_i-\refpt_i)\} + \frac{1}{p}
\sum_{i\notin I} \lambda_i(y_i-\refpt_i) = 1 + \frac{1}{p}\:(k-|I|) \leq 1+\frac{k}{p} \;. \]

Consider any $y'\in\calY\setminus\{y\}$. If there is an index $j\in I$ with $y'_j>y_j=\refpt_j$,
then, since $\calY\subseteq\ints^k$, we know that $y'_j-\refpt_j\geq1$, and therefore
\begin{align*}
	\Norm{y'-\refpt}_p^\lambda &\geq \lambda_j(y'_j-\refpt_j) + \frac{1}{p}\: \lambda_j(y'_j -
\refpt_j) \geq (1+\frac{1}{p})\cdot\lambda_j > 1+k > 1+\frac{k}{p} \;,
\end{align*}
so in this case $y$ is closer to $\refpt$ than $y'$. 

Otherwise, since $y$ is Pareto optimal, there is some $j\in[k]$ such that $y_j<y'_j$, or with
integrality, $y'_j-y_j\geq1$. On the other hand, we know that $y_j-\refpt_j \leq M$. Therefore for
this index $j$,
\[ \lambda_j(y'_j-\refpt_j) = \frac{y'_j-\refpt_j}{y_j-\refpt_j} =
\frac{y'_j-y_j+y_j-\refpt_j}{y_j-\refpt_j} = 1+ \frac{y'_j-y_j}{y_j-\refpt_j} \geq 1+ \frac{1}{M}
\;, \]
and as a consequence
\[ \Norm{y'-\refpt}_p^\lambda \geq \max_{i\in[k]}\{\lambda_i(y'_i-\refpt_i)\} \geq 1+\frac{1}{M} >
1+\frac{k}{p} \;, \]
so again $y$ is closer to $\refpt$ than $y'$.

For the $\ell^p$-norm we let $p>\frac{\log k}{\log(1+\frac{1}{M})}$, and set the weight vector
$\lambda$ as before. We get
\[ \big( \norm{y-\refpt}_p^\lambda \big)^p = \sum_{i\notin I} \left(
\frac{y_i-\refpt_i}{y_i-\refpt_i} \right)^p = k-|I| \leq k \;. \]

If there is a $j\in I$ with $y'_j>y_j=\refpt_j$, then
\[ \big( \norm{y'-\refpt}_p^\lambda \big)^p \geq \big( \lambda_j(y'_j-\refpt_j)\big)^p \geq
\lambda_j^p > k \;. \]
Otherwise, with the same choice of $j\in[k]$ as above,
\[ \big( \norm{y'-\refpt}_p^\lambda \big)^p \geq ( \lambda_j(y'_j-\refpt_j) )^p \geq \left(
1+\frac{1}{M} \right)^p > k \;, \]
where the last inequality holds by the choice of $p$. Thus again in both cases $y$ is closer to
$\refpt$ than $y'$, completing the proof.%\qed
\end{proof}

\paragraph{From approximate Pareto sets to approximating reference point solutions.}
We start the proof of Theorem~\ref{thm:approximation-equiv} by showing that from an
$\alpha$-approximate
Pareto set we can always choose an $\alpha$-approximate solution to $\RPM$.
\begin{lemma}\label{lem:pareto-to-cs} 
	Let $\refpt$ be a feasible reference point, and let $\calY_\alpha$ be an $\alpha$-approximate
Pareto set of $\calP$. Then for any monotone
norm $\norm{\cdot}$, $\min_{y\in\calY_\alpha} \rd(y) \leq \alpha \cdot \min_{y\in\calY}\rd(y)$,
where $\rd(y)=\norm{\refpt} + \norm{y-\refpt}$.
\end{lemma}

\begin{proof}
	Let $\yref \in \calY$ be an optimal solution to $\min_{y\in\calY}\rd(y)$. By monotonicity, we can
w.l.o.g. assume $\yref$ to be Pareto optimal. Thus, there is $y'\in\calY_\alpha$ such that $y' \leq
\alpha \yref$. 
%Then $\rd(\overline{y})\leq \rd(y')$. We thus need to bound $\rd(y')$ from above by
%$\alpha \rd(\ycs)$. 
Using monotonicity and triangle inequality, we get
\[ \norm{y' - \refpt} \leq \norm{ \alpha(\yref - \refpt) + (\alpha-1)\refpt} \leq \alpha\norm{\yref
- \refpt } + (\alpha-1) \norm{\refpt} \;. \]
Reformulation yields
\[ \min_{y\in\calY_\alpha} \rd(y) \leq \rd(y') = \norm{\refpt} + \norm{y' - \refpt} \leq
\alpha (\norm{\refpt} + \norm{\yref - \refpt}) = \alpha \rd(\yref) \;. \qedhere
\]
\end{proof}

\begin{corollary}\label{cor:pareto-rpm}
If there is an $\alpha$-approximation algorithm for the Pareto set of $\calP$, then there is an
$\alpha$-approximation for $\RPM(\calP, \norm{\cdot})$, for every monotone and polynomially
decidable
norm $\norm{\cdot}$.
\end{corollary}

\begin{remark}
Corollary~\ref{cor:pareto-rpm} can also be proved using a result by Mittal and
Schulz~\cite{Mittal2012b}, showing how
to use an
$\alpha$-approximate Pareto set in order to obtain an $\alpha^c$-approximation algorithm for any
monotone low-rank function $h: \calY \rightarrow \reals_{\geq 0}$ fulfilling $h(\mu y) \leq \mu^c
h(y)$ for some constant $c > 0$, all $y \in \calY$, and all $\mu > 1$. Indeed, it can be shown that
the reference point objective function $r$ fulfills this requirement with $c = 1$ for every monotone
norm. Since showing this property is not less effort than proving the result directly, and for
reasons of self-containment, we have included the direct proof here.
\end{remark}

% \medskip
In fact, we can also approximate the compromise solution without knowing the exact ideal point: As
$\calY_\alpha$ contains an $\alpha$-approximate optimal solution for each objective, we can obtain a
reference point $\refpt$ with $\tfrac{1}{\alpha} \ideal \leq \refpt \leq \ideal$. By choosing the
point closest to $\refpt$ from $\calY_\alpha$ we get an $\alpha^2$-approximation to the compromise
solution.

\begin{corollary}\label{cor:pareto-gives-compromise}
If there is an $\alpha$-approximation algorithm for the Pareto set of $\calP$, then there is an
$\alpha^2$-approximation for $\CS(\calP, \norm{\cdot})$, for every monotone and polynomially
decidable norm $\norm{\cdot}$.
\end{corollary}
\begin{proof}
Let $\calY_\alpha$ be an $\alpha$-approximation to the Pareto set. Observe that
$$\refpt_i := \ceil{ \tfrac{1}{\alpha} \min_{y\in \calY_\alpha} y_i }$$
yields a feasible reference point with $\refpt \leq \ideal \leq \alpha \refpt$. 

Now let $y' := \argmin_{y \in \calY_\alpha} \norm{\refpt} + \norm{y - \refpt}$, which by
Corollary~\ref{cor:pareto-rpm} is an $\alpha$-approximation to the reference point solution for
$\refpt$.
Thus, for the compromise solution $\ycs$, we get
$$\norm{y' - \ideal} \leq \norm{y' -\refpt} \leq \alpha\norm{\ycs - \refpt} + (\alpha
-1)\norm{\refpt} \;,$$
where the first inequality follows from monotonicity.
%By this fact and monotonicity, we get
%\begin{align*}
%\norm{y' - \ideal} & \leq \norm{y' - \refpt} \leq \alpha\norm{\ycs - \refpt} +
%(\alpha-1)\norm{\refpt} \;.
%\end{align*}
Observe that $\ideal - \refpt \leq (\alpha - 1)\refpt$ and thus, again by monotonicity,
\begin{align*}
\norm{\ycs - \refpt} \leq \norm{\ycs - \ideal} + (\alpha - 1)\norm{\refpt} \;.
\end{align*}
This finally yields
\begin{align*}
\norm{\ideal} + \norm{y' - \ideal} &\leq \norm{\ideal} +  \alpha(\norm{\ycs - \ideal} + (\alpha -
1)\norm{\refpt}) + (\alpha-1)\norm{\refpt}\\
&\leq \alpha^2 \norm{\ideal} + \alpha\norm{\ycs - \ideal} \;.  \qedhere
\end{align*}
\end{proof}

\paragraph{From approximating reference point solutions to an approximate Pareto set.}

% We continue by showing how to obtain an approximation of the Pareto set by approximating $\RPM$
% for single reference point. To this end, 

In order to show the converse of the result proven above, we use a characterization
from Papadimitriou and Yannakakis~\cite{Papadimitriou2000}, stating that approximability of the
Pareto set is equivalent to
tractability of the so-called \gap\ problem. 
%To include minimization problems and constant factor approximation, we again modify the definition
% given in~\cite{Papadimitriou2000} slightly.

\begin{definition}[\gap\ Problem]
Given an instance of $\calP$ and a vector $y \in \rats^k_{\geq 0}$ as input, the \gap\ problem for
approximation factor $\alpha > 1$, denoted by $\gap(\calP, \alpha)$, is to find a solution $y' \in
\calY$ with $y' \leq y$, or to guarantee that there is no solution $y'' \in \calY$ with $y'' \leq
\frac{1}{\alpha}\:y$.
\end{definition}

%The result from~\cite{Papadimitriou2000} originally is as follows:
%\begin{theorem}[Papadimitriou \& Yannakakis, 2000]\label{thm:fptas-gap}
%There is an FPTAS for the Pareto set of $\calP$ if and only if $\gap(\calP,1+\eps)$ 
%can be solved for all $\eps > 0$ in time polynomial in the size of the input and $\frac{1}{\eps}$.
%\end{theorem}
\begin{theorem}[Papadimitriou \& Yannakakis, 2000]\label{thm:constant-factor-gap}
Let $\calP$ be a multicriteria discrete minimization problem, and let $\alpha>1$. If there is an
$\alpha$-approximation algorithm for the Pareto set, then $\gap(\calP,\alpha)$ is solvable in
polynomial time. If $\gap(\calP,\alpha)$ is solvable in polynomial time, then there is an
$\alpha^2$-approximation algorithm for the Pareto set.
\end{theorem}
%\begin{proof}
%	This follows similarly to \Cref{thm:fptas-gap}.
%\end{proof}

We now show how to use an approximation algorithm for $\RPM$ to solve the \gap\ problem with a
slight increase in the approximation factor. In fact, our result does not even require the algorithm
to solve $\RPM$ for arbitrary reference points. It suffices to find a particular reference point, on
an instance-by-instance basis, that can be approximated. We formalize this by introducing two
algorithms, the first acting as an oracle computing a suitable reference point, which then can be
approximated by the second algorithm.\footnote{Note that it is not sufficient for the first
algorithm to simply return a trivial feasible reference point such as $0$, as it has to ensure that
the second algorithm can provide an approximation for this point. E.g., in the proof of
Corollary~\ref{cor:cs-to-pareto}, it needs to return a point close to the ideal point.}

\begin{lemma}\label{lem:ref-to-pareto}
Let $\alpha > 1$ and set $\beta := \tfrac{\alpha^2}{2\alpha-1}$. There is a polynomial time
algorithm for $\gap(\calP, \alpha)$, if there are two polynomial time algorithms $A_1, A_2$ such
that, 
\begin{itemize}
\item given an instance of $\calP$, algorithm $A_1$ computes a feasible reference point $\refpt \in
\ints^k_{\geq0}$ for that instance, and,
\item additionally given $\refpt$ and $\lambda \in \rats^k_{\geq0}$, algorithm $A_2$ computes in
polynomial time a solution $y' \in \calY$ with $r(y') \leq \beta \min_{z \in \calY} r(z)$, for $r(z)
= \norm{\refpt}_\infty^\lambda + \norm{z - \refpt}_\infty^\lambda$.
\end{itemize}
%Let $\alpha > 1$ and set $\beta := \tfrac{\alpha^2}{2\alpha-1}$. There is a polynomial time
% algorithm for $\gap(\calP, \alpha)$, if
%\begin{itemize}
%\item there is a $\beta$-approximation algorithm for $\CS(\calP, \norm{\cdot}_\infty)$, or
%\item there is a $p_0 \geq 1$ such that there is a family of $\beta$-approximation algorithms for
%$\CS(\calP, \norm{\cdot}_p)$ or $\CS(\calP, \Norm{\cdot}_p)$ for all $p \geq p_0$ with running time
%polynomial in the size of the input and $\log(p)$.
%\end{itemize}
\end{lemma}

\begin{proof}
Let $y \in \rats_{\geq0}^k$ be the input to the $\gap$ problem. W.l.o.g., we can assume that $y \geq
\alpha \refpt$ for the reference point $\refpt$ computed by $A_1$, as otherwise there is no $y' \leq
\tfrac{1}{\alpha}y$ and $\gap$ can be answered negatively. 

We will solve the \gap\ problem with a single call of the $\beta$-approximation algorithm for
$\RPM(\calP, \norm{\cdot}_\infty)$.
For $i \in [k]$, let $\lambda_i := \tfrac{1}{y_i - \refpt_i}$ if $y_i > \refpt_i$, and $\lambda_i :=
2$ if $y_i = \refpt_i = 0$. Let $y'$ be a $\beta$-approximation to $\min_{z \in \calY} r(z)$. 

If $r(y') \leq r(y)$, we return $y'$ as answer to the $\gap$ problem:
$$ 
\lambda_i(y'_i - \refpt_i) \leq \norm{y' - \refpt}_\infty^\lambda \leq \norm{y -
\refpt}_\infty^\lambda \leq 1$$
for all $i \in [k]$ by choice of the weights. Dividing by $\lambda_i$ yields $y'_i \leq y_i$ if $y_i
> 0$, or $y'_i \leq \tfrac{1}{2}$ if $y_i = 0$. In the latter case, integrality of $y'_i$ implies
$y'_i = 0$.

If $r(y') > r(y)$, we answer $\gap$ negatively: Let $y'' \in \calY$. We show that there is an $i \in
[k]$ with $ y''_i > \frac{1}{\alpha} y_i$. First observe that $\beta r(y'') \geq r(y') > r(y)$,
which implies 
$$\beta \norm{y'' - \refpt}_\infty^\lambda > \norm{y - \refpt}_\infty^\lambda - (\beta
- 1) \norm{\refpt}_\infty^\lambda \;.$$ 
Substituting the weights and using $y \geq \alpha\refpt$ yields
$$
\beta \frac{y''_i - \refpt_i}{y_i - \refpt_i} > \frac{y_j - \refpt_j}{y_j - \refpt_j} - (\beta -
1) \frac{\refpt_{j'}}{y_{j'} - \refpt_{j'}} \geq 1 - \frac{\beta - 1}{\alpha - 1} \;,
$$
with $i, j, j'$ being the indices of those components attaining the maxima in the norms. (If either
of the denominators is $0$, then $y''_i > \tfrac{1}{\alpha}y_i$ follows directly.) 
%Note that the second inequality follows from the bound $y \geq \alpha\refpt$ asured in the
%beginning of the proof.
Using the fact that $1 < \beta \leq \alpha$, we get
$$\beta y''_i > (1 - \tfrac{\beta - 1}{\alpha - 1})(y_i - \refpt_i) + \beta \refpt_i \geq (1 -
\tfrac{\beta - 1}{\alpha - 1})y_i \;.$$
It is easy to verify that $\beta = \tfrac{\alpha^2}{2\alpha-1}$ now implies $y''_i >
\tfrac{1}{\alpha}y_i$ and the negative answer to the \gap\ problem is correct.% \qed
%The argumentation for the $\norm{\cdot}_p$- or $\Norm{\cdot}_p$-norm is similar. For the first
%case, an extended version of \Cref{thm:cs-gives-all-pareto} to rational input data ensures that
% $r(y') \geq r(y)$ implies $y' \geq y$ for sufficiently large but polynomially sized $p$. The
% second case is almost identical to the one above.
\end{proof}

% \medskip
As a particular application of Lemma~\ref{lem:ref-to-pareto}, we can show now that also an
approximation
to $\CS$ suffices to approximate the Pareto set:

\begin{corollary}\label{cor:cs-to-pareto}
Let $\alpha > 1$ and set $\beta := \sqrt{\tfrac{\alpha^2}{2\alpha-1}}$. There is a polynomial time
algorithm for $\gap(\calP, \alpha)$, if there is a $\beta$-approximation algorithm for $\CS(\calP,
\norm{\cdot}_\infty)$.
\end{corollary}
\begin{proof}
We show that algorithm $A_1$ and $A_2$ exist, as required by Lemma~\ref{lem:ref-to-pareto}.

Algorithm $A_1$: For every $i \in [k]$, let $\bar{y}^{(i)}$ be a $\beta$-approximation to $\min_{z
\in \calY} r_{\ideal,\bar{\lambda}}(z)$ for weights $\bar{\lambda}_i = 1$ and $\bar{\lambda}_{j} =
0$ for $j \in [k] \setminus \{i\}$. Then $\refpt_i := \ceil{\tfrac{1}{\beta}\bar{y}^{(i)}_i}$
defines a
feasible reference point with $\refpt \leq \ideal \leq \beta \refpt$.

Algorithm $A_2$: Let $\lambda \in \rats^k$. Let $y'$ be a $\beta$-approximation to $\min_{z \in
\calY} r_{\ideal, \lambda}(z)$, and let $\yref = \argmin_{z \in \calY}
r_{\refpt,\lambda}(z)$ be an optimal solution to $\RPM$. We show that $r_{\refpt,\lambda}(y') \leq
\beta^2 r_{\refpt,\lambda}(\yref)$, which concludes the proof.
\begin{align*}
\norm{\refpt}_\infty + \norm{y' - \refpt}_\infty & \leq \norm{\refpt}_\infty + \norm{y' -
\ideal}_\infty + \norm{\ideal - \refpt}_\infty \\
& \leq \beta \norm{\refpt}_\infty + \norm{y' -
\ideal}_\infty\\
& \leq \beta \norm{\refpt}_\infty + \beta \norm{\yref - \ideal}_\infty +
(\beta-1)\norm{\ideal}_\infty\\
& \leq \beta^2 \norm{\refpt}_\infty + \beta \norm{\yref - \ideal}_\infty \;. \qedhere
\end{align*}
\end{proof}

Corresponding versions of Lemma~\ref{lem:ref-to-pareto} and Corollary~\ref{cor:cs-to-pareto} with
the same
approximation factors can be shown for the $\norm{\cdot}_p$- and $\Norm{\cdot}_p$-norms. These
results have been moved to the appendix.

% \medskip
\begin{remark}
	In order to show the result for approximation schemes, let $\alpha=1+\eps$ and $\beta=1+\delta$.
In both cases it suffices to choose $\delta$ such that $1/\delta \in \bigo(1/\eps^2)$, maintaining
polynomiality in $1/\delta$.
\end{remark}

%\paragraph{Finding Compromise Solutions with only Approximate Ideal Points}
%We conclude the proof of \Cref{thm:approximation-equiv} by showing we can solve $\CS$. ...
%\fxnote{fertig machen (evtl. schon weiter oben)}

\section{Approximating Reference Point Solutions}\label{sec:application}
%!TEX root=../mathprog-cs.tex

In this section, we discuss several general techniques for obtaining approximation algorithms for
$\RPM(\mathcal{P}, \norm{\cdot})$. We start with a simple constant factor approximation based on the
weighted sum method, then turn our attention to convex optimization and LP rounding, and close with
approximation schemes arising from pseudopolynomial algorithms.

\paragraph{Approximation by weighted sum.}
Although not all Pareto optimal solutions can be reached by minimizing a weighted sum, this method
still provides an easy way to transfer approximability results from the single-criterion world to
reference point methods.

In~\cite{Diakonikolas2008}, Diakonikolas and Yannakakis show that an approximate \emph{convex}
Pareto set can be computed, if the weighted sum can be optimized. Our results show that via
reference point solutions, also the Pareto set can be approximated. The approximation factor,
however, increases by a factor of $k$.

\begin{theorem}\label{thm:weighted-sum-approx}
If there is an $\alpha$-approximation for $\min_{y \in \calY} \lambda^Ty$, then there is a
$k\alpha$-approximation for $\RPM(\calP, \norm{\cdot}_\infty)$.
\end{theorem}

\begin{proof}
Let $\refpt$ be a feasible reference point and $\lambda \in \rats^k_{\geq0}$. Let $\yref =
\argmin r_{\refpt,\lambda}(y)$ and let $y' \in \calY$ be an $\alpha$-approximation to
$\min_{y \in \calY} \lambda^Ty$. Then
\begin{align*}
\norm{\refpt}^\lambda_\infty + \norm{y' - \refpt}^\lambda_\infty & \leq \norm{\refpt}^\lambda_\infty
+ \lambda^T(y' - \refpt) \\
& \leq \norm{\refpt}^\lambda_\infty + \alpha \lambda^T \yref - \lambda^T \refpt\\
& \leq \norm{\refpt}^\lambda_\infty + \alpha \lambda^T (\yref - \refpt) + (\alpha-1) \lambda^T\refpt\\
& \leq \norm{\refpt}^\lambda_\infty + \alpha k \norm{\yref - \refpt}^\lambda_\infty + (\alpha-1) k \norm{\refpt}^\lambda_\infty \\
& \leq k\alpha (\norm{\refpt}^\lambda_\infty + \norm{\yref - \refpt}^\lambda_\infty) \;. \qedhere
\end{align*}
%The proof for the $\norm{\cdot}_p$- and $\Norm{\cdot}_p$-norms can be found in the appendix.\qed
\end{proof}
In combination with Theorem~\ref{thm:approximation-equiv}, this implies the following result.

\begin{corollary}\label{cor:single-approx-to-pareto-approx}
For any multicriteria combinatorial minimization problem $\calP$ with a constant number of linear
objectives, there is a constant factor approximation for the Pareto set of $\calP$, if and only if
there is a constant factor approximation for the single-criterion version of $\calP$.
\end{corollary}

\paragraph{Convex optimization with linear objectives.}
For optimization problems where the solution space is convex and the objectives are linear (e.g.
linear programming), we can compute reference point solutions w.r.t.~the cornered norm exactly:
\begin{theorem}[Reference point solutions for convex optimization] \label{thm:lp}
	For a multicriteria minimization problem $\min_{x\in\calX} Cx$, with a convex solution set
$\calX\subseteq\reals^n$  for which a polynomial separation algorithm
exists, and a cost matrix $C\in\rats^{k\times n}$, the problem $\min_{x\in\calX}\rd(Cx)$ with
$\rd(y) = \Norm{\refpt}_p + \Norm{y-\refpt}_p$, for any feasible reference point $\refpt$ and any
$p\in[1,\infty]$, is again a convex optimization problem with linear objectives and thus solvable in
polynomial time.
\end{theorem}
\begin{proof} The problem can be formulated as follows:
	\begin{align*}
		\min_{x\in\calX}{\rd}(Cx) = \norm{\refpt}_\infty + \left\{ 
		\begin{array}{ll}
			\min\quad &\Delta + \frac{1}{p}\cdot \vecofones\transp Cx \\[.5em]
			\text{s.t.} & Cx - \refpt \leq \Delta\cdot\vecofones \\
			&x\in\calX \\
			&\Delta\in\reals \;.
		\end{array} \right.
	\end{align*}
	Here, $\vecofones$ denotes the vector of ones of corresponding dimension. In the optimum, $\Delta
=
\max_i \{ c_{i\cdot}x - \refpt_i \}$, and therefore the two programs are equivalent. The objective
is clearly linear, and the solution space is $\calX \times \reals$, intersected with the halfspaces
defined by the inequalities $c_{i\cdot}x - \refpt_i \leq \Delta$, $i\in[k]$, and thus convex.

Since we can solve the separation problem for the original set $\calX$, we can also solve it for
the set with the added inequalities. By the equivalence of separation and optimization
(Gr\"otschel et al.~\cite{Groetschel1981}) we can solve $\min_{x\in\calX}{\rd}(Cx)$ in
polynomial
time.%\qed
\end{proof}

\begin{remark} A special case of convex optimization problems are linear programs (LPs).
From our result it follows that we can exactly compute reference point solutions for multicriteria
LPs. It also yields a nice alternative proof of the existence of an FPTAS for the Pareto set, which
has first been proven in Papadimitriou and Yannakakis~\cite{Papadimitriou2000} using an involved
geometric argument. 

A different argument for the approximability of Pareto sets of linear programs has independently
been noted by Mittal and Schulz~\cite{Mittal2012}.
\end{remark}

\begin{corollary}\label{cor:approx-for-lp}
Let $\mathcal{P}$ be a multicriteria minimization problem with convex feasible set and linear
objective functions. Assume that there is a positive polynomial $\pi$ such that $\calY \subseteq \{y
\in
\rats^k:~ y_i \geq \tfrac{1}{\pi(|I|)}\ \forall i \in [k]\}$, where $|I|$ is the encoding length of
the instance. If there is a polynomial time algorithm for the separation problem of $\mathcal{P}$,
then there is an FPTAS for the Pareto set. 
%  For multicriteria minimization problems with convex separable solution sets containing only
% strictly positive objective vectors and linear objectives there is an FPTAS for the Pareto set.
\end{corollary}
%\paragraph{Remark.}  \fxnote{move to intro} This result is also a consequence of the fact that the
% \gap\ corresponding to an LP is again an LP -- this follows similarly to the proof above. 
\begin{remark}[Convex sets and the integrality assumption]
Note that our general integrality assumption $\calY \subseteq \ints_{\geq0}^k$ for discrete
optimization problems, introduced in Section \ref{sec:preliminaries}, does not hold for the case of
convex
optimization problems in Theorem~\ref{thm:lp} and Corollary~\ref{cor:approx-for-lp}. However, by
assuming $y_i \geq
\tfrac{1}{\pi(|I|)}$ for all occurring objective values in Corollary~\ref{cor:approx-for-lp}, we
ensure that
all prerequisites stated in Papadimitriou and Yannakakis~\cite{Papadimitriou2000} for
Theorem~\ref{thm:constant-factor-gap} are
still
fulfilled. Furthermore observe that, while our proof of Lemma~\ref{lem:ref-to-pareto} also assumed
integral objectives, we used this integrality assumption only for showing that if the solution $y'$
computed by algorithm $A_2$ fulfills $r(y') \leq r(y)$, then $y_i = 0$ implies $y'_i = 0$. However,
we can ignore this case, as by our assumption all objectives are strictly positive and thus $y_i =
0$ already implies that the answer to \gap\ is negative. Thus, both Lemma~\ref{lem:ref-to-pareto}
and
Theorem~\ref{thm:constant-factor-gap} are still valid for convex optimization problems fulfilling
the
condition of Corollary~\ref{cor:approx-for-lp}.
\end{remark}

% \medskip
\begin{proof}[Proof of Corollary~\ref{cor:approx-for-lp}]
By Theorem~\ref{thm:lp}, we can compute an optimal solution to $\RPM(\calP, \norm{\cdot}_\infty)$
for any
reference point in polynomial time. Thus, by Lemma~\ref{lem:ref-to-pareto}, we can solve
$\gap(\calP,
1+\varepsilon)$ in polynomial time for any $\varepsilon > 0$ (with running time independent of
$\varepsilon$), which by Theorem~\ref{thm:constant-factor-gap} gives an FPTAS for the Pareto set.
\end{proof}

\paragraph{Approximation through LP rounding.}
One of the most successful techniques for the design of approximation algorithms for
integer problems is \emph{LP rounding}: The problem is formulated as a linear integer program (IP),
then the integrality constraints are relaxed and the resulting LP is solved, and finally the
optimal fractional solution is rounded to a feasible integral solution, losing only a certain
factor in the objective.

Many important LP rounding algorithms are \emph{oblivious} in the sense that the rounding procedure
is independent of the cost function. We show that these algorithms can be adapted such that they
also solve the reference point version of the problem, with the same approximation factor.
\begin{theorem}\label{thm:lp-rounding}
	Consider a multicriteria minimization problem $\min_{x\in\calX} Cx$ with a solution
set $\calX\subseteq\ints_{\geq0}^n$ and a cost matrix $C\in\rats^{k\times n}$. If there exist 
\begin{itemize}
	\item a convex relaxation $\calX'$ for which the separation problem can be solved in polynomial
time, 
	\item and a polynomial time rounding procedure $\calR: \calX' \rightarrow \calX$ such that
$c\transp\calR(x') \leq \alpha c\transp x'$ for all $c\in\rats^n_{\geq0}$ and all $x'\in\calX'$,
\end{itemize}
then for any feasible reference point $\refpt$ and any $p\in[1,\infty]$ there is an
$\alpha$-approximation algorithm for $\min_{x\in\calX} \rd(Cx)$, with
$\rd(y)=\Norm{\refpt}_p+\Norm{y-\refpt}_p$.
\end{theorem}
\begin{proof}
	From Theorem~\ref{thm:lp} it follows that we can compute in polynomial time a fractional
solution $x'\in\calX'$ minimizing ${\rd}(Cx)$. Let $x=\calR(x')$. Then
\begin{align*}
	{\rd}(Cx) &= \max_{i\in[k]} \{\refpt_i\} + \max_{i\in[k]} \{ (Cx)_i - \refpt_i \} +
\frac{1}{p} \cdot \sum_{i\in[k]} (Cx)_i \\
	&\leq \max_{i\in[k]} \{\refpt_i\} + \max_{i\in[k]} \{ \alpha(Cx')_i - \refpt_i \} + \alpha \cdot
\frac{1}{p} \cdot \sum_{i\in[k]} (Cx')_i \\
	&= \max_{i\in[k]} \{\refpt_i\} + \max_{i\in[k]} \big\{ \alpha\big( (Cx')_i - \refpt_i \big) +
(\alpha-1) \refpt_i \big\} + \alpha \cdot
\frac{1}{p} \cdot \sum_{i\in[k]} (Cx')_i \\
	&\leq \alpha \cdot \max_{i\in[k]} \{\refpt_i\} + \alpha \cdot \max_{i\in[k]} \big\{ (Cx')_i -
\refpt_i
\big\} + \alpha \cdot \frac{1}{p} \cdot \sum_{i\in[k]} (Cx')_i \\
	&= \alpha \cdot {\rd}(Cx') \;. \qedhere
\end{align*}
\end{proof}

Theorem~\ref{thm:lp-rounding} immediately results in the approximability, with a factor independent
of~$k$, of reference
point solutions and the Pareto set for several classical combinatorial optimization problems. We
give two examples here.
%Note that in these problems we can not compute the ideal point exactly, since the unicriteria
% problem is already NP-hard. However, if we have an $\alpha$-approximation algorithm for the
% unicriteria problem then we use this to compute approximations $y_i^\alpha$ of the coordinates of
% the ideal point and set $\refpt_i = \frac{1}{\alpha}y_i^\alpha$. This gives a feasible reference
% point $\refpt$ with $\frac{1}{\alpha}\ideal \leq \refpt \leq \ideal$.

For \textsc{Set Cover}, in 1982 Hochbaum~\cite{Hochbaum1982} presented an LP-based
$\kappa$-approximation algorithm, where $\kappa$ is the maximum cardinality of a set. Thus, there is
a $\kappa$-approximation algorithm for the corresponding reference point version, and a
$\bigo(\kappa^2)$-approximation algorithm for the Pareto set. A notable special case is
\textsc{Vertex
Cover}, where $\kappa=2$.

For the scheduling problem of minimizing the weighted sum of completion times on a single machine
with release dates ($1|r_j|\sum w_jC_j$), Hall et al.~\cite{Hall1997} gave a $3$-approximation
algorithm based on an LP-relaxation, resulting in a $3$-approximation for compromise solutions,
which gives a constant factor approximation for the Pareto set as well. M\"ohring et
al.~\cite{Moehring1999} extended this to stochastic scheduling with random processing times
($P|p_j\sim \text{stoch}, r_j|E[\sum w_jC_j]$), for which we consequently also get constant factor
approximations for the multicriteria problems.

\begin{remark} While we usually restrict ourselves to the case of a constant number of
criteria, the results on convex optimization and LP-rounding also hold for a polynomial number of
criteria. This is due to the fact that we can still solve the linear program if we add a polynomial
number of constraints.
\end{remark}

\paragraph{From pseudopolynomial algorithms to approximation schemes.}%\label{subsec:aissi}

Multicriteria optimization, and in particular the concept of compromise solutions, is closely
related
to robust optimization. If each criterion is considered as one \emph{scenario} in the robust
setting, then a compromise solution w.r.t.~$\norm{\cdot}_\infty$ is exactly the same as a
\emph{min-max regret robust} solution.

Aissi et al.~\cite{Aissi2006} consider this robust setting for binary optimization problems and show
that if we can compute upper and lower bounds
on the optimum which only differ by a polynomial factor, and if there is a
pseudopolynomial algorithm whose running time depends on the encoding length of the instance and the
upper
bound, then there is an FPTAS for the min-max regret robust problem. We show that this result can be
extended to reference point solutions.

%\paragraph{The Framework.}
%The results in this section only hold for binary optimization problems, i.e., the set of feasible
%solutions is $\calX\subset\{0,1\}^n$. The costs are given by a matrix $C\in\nats^{k\times n}$, i.e.
%the cost of a solution $x\in\calX$ w.r.t. criterion~$i$ is $c_i x$.

\newcounter{count-thm:aissi-ext}
\addtocounter{count-thm:aissi-ext}{\arabic{theorem}}
\begin{theorem}\label{thm:aissi-ext}
	Consider a multicriteria minimization problem $\min_{x\in\calX}Cx$ with a set  of feasible
solutions
$\calX\subseteq\{0,1\}^n$ and cost matrix $C\in\ints_{\geq0}^{k\times n}$. For any $p \in [1,
\infty]$, if
	\begin{enumerate}
		\item for any instance $I=(\calX,C)$, and any feasible reference point $\refpt$, a lower and an
upper bound $L$ and $U$ on
$\min_{x\in\calX} \rd(Cx)$ can be computed in time $\pi_1(|I|)$, such that $U\leq
\pi_2(|I|)L$, where $\pi_1$ and $\pi_2$ are non-decreasing polynomials,
		\item and there exists an algorithm that solves $\min_{x\in\calX} \rd(Cx)$ for any
instance $I=(\calX,C)$ in time $\pi_3(|I|,U)$, where $\pi_3$ is a non-decreasing polynomial,
	\end{enumerate}
	then there is an FPTAS for $\min_{x\in\calX} \rd(Cx)$, where $\rd(y) =
\Norm{\refpt}_p + \Norm{y-\refpt}_p$.
%\fxnote{state $(4+\eps)$-approx for $\ell^p$-norm?}
\end{theorem}
By $|I|$ we denote the encoding length of the instance $I$.

\begin{proof}
	\newcommand{\refptbar}{\overline{y}^{\textup{rp}}}
	To compute a $(1+\eps)$-approximation to the reference point solution, we set
$\eps'=\eps\cdot(1+\frac{k}{p})^{-1}$ and apply the pseudopolynomial algorithm to a modified
instance
$\overline{I}$ with cost coefficients $\overline{c}_{ij}:=\floor{\frac{3n}{\eps'L} \, c_{ij}}$.
Observe that
\[ \frac{\eps' L}{3n}\cdot \overline{c}_{ij} \leq c_{ij} < \frac{\eps' L}{3n}(\overline{c}_{ij}+1)
 \;. \]
The reference point for the modified instance is defined by $\refptbar_i := \floor{
\frac{3n}{\eps'L} \, \refpt_i }$. This reference point is feasible for the modified instance, and it
holds that
\begin{equation} 
	\frac{\eps'L}{3n} \, \refptbar_i \leq \refpt_i < \frac{\eps'L}{3n} \refptbar_i +
\frac{\eps'L}{3n} < \frac{\eps'L}{3n} \, \refptbar_i + \frac{\eps'L}{3} \;. \label{eq:refpt-bound}
\end{equation}

Let $x^*$ and $\overline{x}^*$ be reference point solutions for $I$ and $\overline{I}$,
respectively. We now bound the value of $\overline{x}^*$ w.r.t.~the original costs $c$. Let $\rd$
and $\overline{\rd}$ denote the reference point objective function for the original and the modified
costs,
respectively. We get
\begin{align*}
	\rd(C\overline{x}^*) &= \Norm{\refpt}_p + \Norm{C\overline{x}^* - \refpt}_p \\
	&\leq \frac{\eps'L}{3n}\Norm{\refptbar}_p + \frac{\eps'L}{3}\left(1+\frac{k}{p}\right) \\
	&\quad+ \max_{i\in[k]} \left\{ \frac{\eps'L}{3n} (\overline{c}_i \overline{x}^* -
\refptbar_i) \right\} + \frac{\eps' L}{3} + \frac{1}{p}\sum_{i\in[k]} \left(\frac{\eps'L}{3n}
(\overline{c}_i \overline{x}^* - \refptbar_i) + \frac{\eps' L}{3} \right) \\
	&= \frac{\eps'L}{3n}\Norm{\refptbar}_p + \frac{\eps'L}{3n} \Norm{ \overline{C}\overline{x}^*
- \refptbar_i }_p + \frac{2\eps'L}{3}\left(1+\frac{k}{p}\right) \\
	&\leq \frac{\eps'L}{3n} \Norm{\refptbar}_p + \frac{\eps'L}{3n} \Norm{ \overline{C}x^* -
\refptbar_i }_p + \frac{2\eps'L}{3}\left(1+\frac{k}{p}\right) \\
	&\leq \Norm{\refpt}_p + \frac{\eps'L}{3n} \left( \frac{3n}{\eps'L} \Norm{Cx^* - \refpt}_p +
n\Big(1+\frac{k}{p}\Big) \right) + \frac{2\eps'L}{3}\left(1+\frac{k}{p}\right) \\
	&= \Norm{\refpt}_p + \Norm{ Cx^* - \refpt}_p + \eps'L\left(1+\frac{k}{p} \right) \\
	&= \rd(Cx^*) + \eps L \\
	&\leq (1+\eps) \rd(Cx^*) \;.
\end{align*}

It remains to be shown that $\overline{x}^*$ can be computed in time polynomial in $|I|$ and
$\frac{1}{\eps}$. For this, denote by $\overline{L}$ and $\overline{U}$ the lower and upper
bounds on the optimal value $\overline{\opt}$ of the modified instance $\overline{I}$. According to
the prerequisites of the theorem, we can compute $L$ and then $\overline{x}^*$ in time
\begin{align*}
\pi_1(|I|) + \pi_3(|\overline{I}|,\overline{U})
&\leq \pi_1(|I|) + \pi_3(|\overline{I}|,\pi_2(|\overline{I}|)\overline{L}) \\
&\leq \pi_1(|I|) + \pi_3(|\overline{I}|,\pi_2(|\overline{I}|)\overline{\opt}) \\
&\leq \pi_1(|I|) + \pi_3\left( |\overline{I}|, \pi_2(|\overline{I}|) \big(
\tfrac{3n}{\eps'}\pi_2(|I|) + n(1+\tfrac{k}{p}) \big) \right) \;,
\end{align*}
where the last inequality holds because
\begin{align*}
	\overline{\opt} &\leq \frac{3n}{\eps'L}\Norm{\refpt}_p + \frac{3n}{\eps'L} \Norm{Cx^* - \refpt}_p
+ n\Big(1+\frac{k}{p}\Big) \\
	&\leq \frac{3n}{\eps'L}\cdot U + n\Big(1+\frac{k}{p}\Big) \\
	&\leq \frac{3n}{\eps'}\cdot \pi_2(|I|) + n\left(1+\frac{k}{p}\right) \;. \\
\end{align*}
Finally note that $|\overline{I}|\leq \pi_4(|I|,\log\frac{1}{\eps},\log\frac{k}{p})$ for some
polynomial $\pi_4$.
Thus the above calculations prove that the running time is indeed polynomial.%\qed
\end{proof}

% \medskip
\begin{remark}
	Theorem~\ref{thm:aissi-ext} also holds for $\CS(\calP,\Norm{\cdot}_p)$.
\end{remark}
\begin{proof}
  \newcommand{\idealbar}{\overline{y}^{\textup{id}}}
For compromise solutions, we can not choose the reference point of the modified instance as we see
fit. However, also for the ideal point, Equation~\eqref{eq:refpt-bound} still holds. To see this,
denote the
respective ideal points by $\ideal$ and $\idealbar$, and let $x^{(i)},\overline{x}^{(i)}$ for
$i\in[k]$ be optimal solutions of $\min_{x\in\calX} c_i x$ and $\min_{x\in\calX} \overline{c}_i x$.

It holds that
\begin{align*}
	\ideal_i &=	c_i x^{(j)}\geq c_i \overline{x}^{(j)} \geq \frac{\eps' L}{3n}
\overline{c}_i \overline{x}^{(j)} = \frac{\eps' L}{3n} \idealbar_i \;, \\
	\ideal_i &=	c_i x^{(j)} \leq \frac{\eps' L}{3n}(\overline{c}_i+\vecofones\transp) x^{(j)}
\leq \frac{\eps' L}{3n} \overline{c}_i x^{(j)} + \frac{\eps' L}{3} \leq \frac{\eps' L}{3n}
\idealbar_i + \frac{\eps' L}{3} \;,
\end{align*}
so Equation~\eqref{eq:refpt-bound} also holds for the ideal points.
\end{proof}

% \medskip
\begin{remark}
For the running time, it is essential that $p$ is fixed, or at least bounded from below by a positive
constant (e.g. $p\geq1$), as the running time is only polynomial in $\frac{1}{p}$. Since for
$p\rightarrow0$ compromise programming becomes equivalent to the weighted sum
problem, this is only a minor restriction.
\end{remark}

% \medskip
Similarly to Proposition 1 in Aissi et al.~\cite{Aissi2006}, we can show that the necessary bounds
$U$ and $L$
can
be computed, if the single-objective problem is tractable. This is a direct implication of the
weighted sum approximation described in Theorem~\ref{thm:weighted-sum-approx}.

\begin{corollary}
If there is an $\alpha$-approximation for the single-criterion version of $\calP$, then for all
instances of $\RPM(\calP, \norm{\cdot})$, we can compute $L$ such that $L\leq \min_{y\in \calY}
r(y) \leq \alpha kL$.
\end{corollary}

The pseudopolynomial algorithms for the shortest path problem (SP) and the minimum spanning tree
problem (MST) presented in Aissi et al.~\cite{Aissi2006} can be used to compute reference point
solutions as
well, as
they both compute all (non-dominated) regret vectors (that obey the upper bound $U$), and the
reference point solution always has a non-dominated regret vector.

%The result for MST uses a translation of the reference point to the origin, a transformation that
% can also be applied in the context of reference point solutions.

\begin{corollary}%\fxnote{check whether compromise solution or general reference point}
There is an FPTAS for $\RPM(\textup{SP}, \Norm{\cdot}_p)$ and $\RPM(\textup{MST}, \Norm{\cdot}_p)$
for any $p\in[1,\infty]$.
\end{corollary}

\section{Maximization}\label{sec:maximization}
We now investigate which of the results from Section~\ref{sec:general} hold for maximization
problems. Note
that for the problem $\max_{y\in\calY}y$, the ideal point $\ideal$ is defined by $\ideal_i =
\max_{y\in\calY} y_i$, and a reference point $\refpt\in\ints_{\geq0}^k$ is called feasible if
$\refpt\geq\ideal$. A solution $y\in\calY$ is Pareto optimal if there is no
$y'\in\calY\setminus\{y\}$ with
$y'\geq y$. An $\alpha$-approximate Pareto set has to contain, for all $y\in\calP$, a solution $y'$
with $y'\geq\frac{1}{\alpha}\,y$. Accordingly, feasible answers to $\gap(P,\alpha)$ for an input
vector $y$ are either a vector $y'\in\calY$ with $y'\geq y$, or the guarantee that there is no
vector $y''\in\calY$ with $y''\geq \alpha y$.

Recall that the objective function for compromise and reference point
solutions is the value of the reference point, degraded by the price of
compromise. For maximization, we have to substract the price of
compromise, i.e., $\rd(y) = \norm{\refpt} - \norm{\refpt - y}$. 
This objective function is then aimed to be maximized.
To simplify the presentation, in this section we restrict to
statements about monotone norms and about the infinity-norm.

We begin our considerations with an observation.
\begin{observation}%\label{obs:max-weighted-sum-to-pareto}
  It is \emph{not} true that whenever there is a constant factor approximation algorithm for the
weighted sum problem $\max_{y\in\calY} \lambda\transp y$, then there also is an approximation
algorithm for the Pareto set.
\end{observation}

% \smallskip
\begin{proof}
  Suppose there is an $\alpha$-approximation algorithm for $\max_{y\in\calY} \lambda\transp y$.
Consider an instance with $k=2$ and $\calY = \{ (1,1), (3,0), (0,3) \}$. For any
$\lambda\in\rats_{\geq0}^k$, the approximation algorithm could either return $(3,0)$ or $(0,3)$.
Hence an algorithm that only relies on the existence of a weighted sum approximation can not tell
whether the element $(1,1)$ exists or not. Any approximate Pareto set, however, has to contain the
point $(1,1)$, if it exists.
\end{proof}

% \smallskip
This shows that at least one of the implications of approximability depicted in
Figure~\ref{fig:equivalence-graph} no longer holds for maximization problems. Some of them, however,
continue to hold. We get analogues of Lemma~\ref{lem:pareto-to-cs} and
Corollaries~\ref{cor:pareto-rpm} and \ref{cor:pareto-gives-compromise}, implying the approximability
of compromise respectively reference point solutions, in case the Pareto set is approximable.
\begin{lemma} \label{lem:max_pareto-to-cs2}
  Let $\refpt$ be a feasible reference point for $\max_{y\in\calY}y$, and let $\calY_\alpha$ be an
$\alpha$-approximate Pareto set. Then for any monotone norm $\norm{\cdot}$, $\max_{y\in\calY_\alpha}
\rd(y)  \geq \frac{1}{\alpha} \cdot \max_{y\in\calY}\rd(y)$, where $\rd(y)=\norm{\refpt} -
\norm{\refpt-y}$.
\end{lemma}
\begin{proof}
  Let $\yref \in \calY$ be an optimal solution to $\min_{y\in\calY}\rd(y)$, and let
$y'\in\calY_\alpha$ with $y' \geq \frac{1}{\alpha} \yref$. Then
\[ \norm{\refpt - y'} \leq \frac{1}{\alpha} \norm{\refpt - \yref} + \Big(1-\frac{1}{\alpha}\Big)
\norm{\refpt} \;, \]
and hence
\[ \max_{y\in\calY_\alpha} \rd(y) \geq \norm{\refpt} - \norm{\refpt - y'} \geq
\frac{1}{\alpha} (\norm{\refpt} - \norm{\yref - \refpt}) \;. \qedhere \]
\end{proof}

\begin{corollary}\label{cor:max_pareto-rpm2}
If there is an $\alpha$-approximation algorithm for the Pareto set of $\max_{y\in\calY}y$, then
there is an $\alpha$-approximation for $\max_{y\in\calY} \norm{\refpt} - \norm{\refpt-y}$, for every
monotone and polynomially decidable norm $\norm{\cdot}$.
\end{corollary}

\begin{corollary}\label{cor:max_pareto-gives-compromise2}
If there is an $\alpha$-approximation algorithm for the Pareto set of $\max_{y\in\calY}y$, then
there is an $\alpha^2$-approximation for $\max_{y\in\calY} \norm{\ideal} - \norm{\ideal-y}$, for
every monotone and polynomially decidable norm $\norm{\cdot}$.
\end{corollary}

Interestingly, in the reverse direction, compromise programming and reference point methods are 
suddenly of different complexities: For compromise solutions, there is no analogue of 
Corollary~\ref{cor:cs-to-pareto}:
\begin{observation}
  For any $0<\eps<1$ and sufficiently large $M$, there are instances $\calY(M,\eps)$ and
$\calY'(M,\eps)$, with encoding length $\bigo(\log M)$, that have different
$(1+\eps)$-approximate Pareto sets, but a $(1+\delta)$-approximation algorithm for $\max_{y\in\calY}
\norm{\ideal}^\lambda_\infty - \norm{\ideal - y}^\lambda_\infty$ can only distinguish between the
two instances for $\delta \in \bigo(1/M)$.
\end{observation}

% \smallskip
\begin{proof}
  Let $\rd_\lambda(y) := \norm{\ideal}^\lambda_\infty + \norm{\ideal - y}^\lambda_\infty$. Since
  approximation is invariant under scaling of $\lambda$, we can assume w.l.o.g.~that 
  $\norm{\lambda}_\infty = 1$. Consider the two sets
  \begin{align*}
    \calY(M,\eps) &= \left\{ y=\binom{1}{M+1}, y'=\binom{(1+\eps)^{-2}}{M+(1+\eps)^{-2}},
y''=\binom{1}{\frac{M}{2}+1}, \binom{M+1}{0}, \binom{0}{2M+1} \right\}\;, \\
    \calY'(M,\eps) &= \calY(M,\eps) \setminus \{y\} \;.
  \end{align*}
  %The ideal point of these two instances is $\ideal = (M+1,2M+1)$. 
  For sufficiently large values of $M$, a $(1+\eps)$-approximate Pareto set of $\calY$ has to 
  contain $y$. However, for any $\lambda\in\rats_{\geq0}^2$, either $y'$ (for $\lambda_2\geq2/3$) 
  or $y''$ (for $\lambda_2\leq2/3$) is a $(1+\delta)$-approximation to $r_\lambda(y)$, unless 
  $\delta\in\bigo(1/M)$.
\end{proof}

% \smallskip
If the reference point can be chosen freely, however, we do get an analogue of 
Lemma~\ref{lem:ref-to-pareto}:
\begin{lemma}
  Let $\calP$ be a multicriteria maximization problem, and let 
  $\alpha > 1$. There is a polynomial time algorithm for $\gap(\calP, \alpha)$, 
  if for any feasible reference point $\refpt$ and any $\lambda\in\rats^k_{\geq0}$ there is an
  $\alpha$-approximation algorithm for $\max_{z \in \calY} r_\lambda(z)$, where $r_\lambda(z) = 
  \norm{\refpt}_\infty^\lambda - \norm{\refpt - z}_\infty^\lambda$.
\iffalse
  if there are two polynomial time algorithms $A_1, A_2$ such that, 
\begin{itemize}
\item given an instance of $\calP$, algorithm $A_1$ computes a feasible reference point $\refpt \in
\rats^k_{\geq0}$ for that instance, and,
\item additionally given a reference point $\modrp\geq\refpt$ and $\lambda \in \rats^k_{\geq0}$,
algorithm $A_2$ computes in polynomial time a solution $y' \in \calY$ with $r_\lambda(y') \geq 
\frac{1}{\alpha} \max_{z \in \calY} r_\lambda(z)$ for $r_\lambda(z) = \norm{\modrp}_\infty^\lambda
- 
\norm{\modrp - z}_\infty^\lambda$.
\end{itemize}
\fi
\end{lemma}
\begin{proof}
  Let $y\in\rats^k$ be the input to the \gap\ problem, w.l.o.g.~$y\neq0$.
  %If $y=0$, then any feasible solution is a positive answer to the \gap\ problem. 
  % We therefore assume $y\neq0$. 
  Let $I=\{i\in[k]: y_i\neq0\}$, and let $M$ be an upper bound on the objective values.
  Further, let $c=\max_{i\in I} \frac{M}{y_i}$, and set $\refpt = c\cdot y$. Note that this is a 
  feasible reference point.
%We first compute a new feasible reference point $\modrp$ with $\modrp_i = c y_i$ for all $i\in I$
%for some constant $c$. For this, let 
%  \[ c= \max_{i\in I} \frac{\refpt_i}{y_i} \qquad\text{and}\qquad \modrp = c\cdot \refpt \;. \]
%  Clearly, $\modrp \geq \refpt \geq \ideal$, so the new reference point is feasible. 
We now set the weight vector to
  \[ \lambda_i = \begin{cases} 1/y_i &\text{for } i\in I, \\ 0 &\text{otherwise.} \end{cases} \]
  Let $y'$ be an $\alpha$-approximate solution to $\max_{z\in\calY} r_\lambda(z)$.
  
  If $r_\lambda(y') \geq r_\lambda(y)$, then $y'_i\geq y_i$ for all $i\in I$, and $y'_i \geq0=y_i$ for all other
$i$, so $y'$ is a positive answer to the \gap\ problem. Otherwise, i.e.~if $r_\lambda(y')<r_\lambda(y)$, for any
$y''\in\calY$, we know $r_\lambda(y'')\leq\alpha r_\lambda(y') <\alpha r_\lambda(y)$. Let $j = \arg\max_{i\in I}
\frac{\refpt_i - y''_i}{y_i}$. Then,
  \begin{align*}
  r_\lambda(y) &= c\norm{y}^\lambda_\infty - (c-1)\norm{y}^\lambda_\infty = \norm{y}^\lambda_\infty = 1
\\[.5em]
  \Rightarrow\quad \alpha &= \alpha r_\lambda(y) > r_\lambda(y'') = c\cdot\norm{y}^\lambda_\infty - \norm{\refpt -
y''}^\lambda_\infty = c - \frac{cy_j-y''_j}{y_j} \\
  \Rightarrow\quad y''_j &< \alpha y_j - c y_j + c y_j = \alpha y_j \;.
\end{align*}
We can therefore conclude that there is no $y''\in\calY$ with $y''\geq\alpha y$, and answer the
\gap\ problem negatively.
\end{proof}

% \smallskip
All approximability reductions for compromise and reference point solutions for maximization
problems are depicted in Figure~\ref{fig:max_graph-of-equivalences} below.

\begin{figure}[ht]
  \begin{center}
		\begin{tikzpicture}
    [xscale=1.2,yscale=1.5,textnode/.style={rounded corners,draw=black,inner sep=5pt,align=left}]
      \footnotesize   
      \node[textnode](pareto) at (0,1) {Pareto set};
      \node[textnode](gap) at (2,1) {$\gap(y,\alpha)$\\$\forall\;y\in\rats^k$};
      \node[textnode](ws) at (4,1) {$\min_y \lambda\transp y$};
      \node[textnode](cs-gen) at (0,2) {$\CS(\norm{\cdot})$};
                                                % ,\\ $\norm{\cdot}$ monotone\\\& poly decidable};
  %     \node[textnode](cs-p) at (2,2) {$\CS(\norm{\cdot}_p)$ and\\ $\CS(\Norm{\cdot}_p)$,
%   $p\geq1$};
      \node[textnode](cs-inf) at (4,2) {$\CS(\norm{\cdot}_\infty)$};
      \node[textnode](rpm-gen) at (0,0) {$\RPM(\norm{\cdot})$};
      \node[textnode](rpm-inf) at (4,0) {$\RPM(\norm{\cdot}_\infty)$};
      \draw[<->] (pareto) -- (gap);
      \draw[->] (gap) -- (ws);
      \draw[->] (cs-gen) -- node[above]{} (cs-inf);
  %     \draw[->] (cs-p) -- node[above]{} (cs-inf);
      \draw[->] (pareto) -- node[above]{} (cs-gen);
  %     \draw[->] (cs-p) -- node[above]{} (gap);
%       \draw[->] (cs-inf) -- node[above]{} (gap);
  %     \draw[dashed,->] (ws) -- node[right]{add.~factor $k$} (cs-inf);
%   
      \draw[->] (rpm-gen) -- node[above]{} (rpm-inf);
      \draw[->] (pareto) -- node[above]{} (rpm-gen);
      \draw[->] (rpm-inf) -- node[above]{} (gap);
    \end{tikzpicture}
		\caption{Reductions of approximability for maximization
problems.\label{fig:max_graph-of-equivalences}}
  \end{center}
\end{figure}
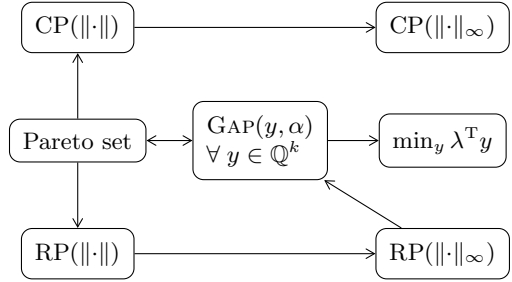

\section{Conclusion}

A multicriteria optimization problem lacks a single, unifying objective function. A priori, there is
no metric justifying a preference on the set of Pareto optimal solutions. Still, the Pareto
solutions are not equivalent like, e.g., the set of optima for a single criterion. Decision makers
can have preferences among the Pareto solutions. Reference point
methods model such preferences. These methods are widespread in practice
and are a more powerful model than a simple weighing of the objectives. They are the most powerful
model in the sense that every Pareto solution can become the unique optimum, for some choice of the
additional input. 
To the
best of our knowledge, this paper provides the first extensive theoretical study of these methods in
the context of approximation. 

Our main results establishes computational equivalence between the
approximation of the Pareto set and the approximation of reference
point solutions, thus linking the rich body of mathematical research
on Pareto sets to the practically widespread reference point methods.
Moreover, this work lifts a number of
important and general algorithmic techniques known for single criteria
optimization to the setting of reference point solutions.

%Our study focused exclusively on minimization problems, and we note that this is \emph{not} without loss of generality. A theoretical investigation of reference point methods for maximization problems will be topic of future research.

\paragraph{Acknowledgment.} We are grateful to G\"unter
Ziegler for a discussion that significantly simplified the proof of \ref{lem:pareto-to-cs}.

\bibliographystyle{plain}
\bibliography{cs}

\newpage
\section*{Appendix}

% \spnewtheorem{thmrevisited3}[count-lem:ref-to-pareto]{Lemma}{\bfseries}{\itshape}
\begin{lemma}[Lemma~\ref{lem:ref-to-pareto} revisited, for $\Norm{\cdot}_p$ and $\norm{\cdot}_p$]
\label{lem:ref_p-to-pareto}
Let $\alpha > 1$ and set $\beta := \tfrac{\alpha^2}{2\alpha-1}$. There is a polynomial time
algorithm for $\gap(\calP, \alpha)$, if there are two polynomial time algorithms $A_1, A_2$ such
that, 
\begin{itemize}
\item given an instance of $\calP$, algorithm $A_1$ computes in polynomial time a feasible reference
point $\refpt \in \ints^k_{\geq0}$ for that instance, and,
\item additionally given $\refpt$ and $\lambda \in \rats^k_{\geq0}$ and $p \geq 1$, algorithm $A_2$
computes in polynomial time a solution $y' \in \calY$ with $r(y') \leq \beta \min_{z \in \calY}
r(z)$, for $r(z) = \Norm{\refpt}_p^\lambda + \Norm{\refpt - z}_p^\lambda$ or $r(z) =
\norm{\refpt}_p^\lambda + \norm{\refpt - z}_p^\lambda$, respectively.
\end{itemize}
\end{lemma}

\begin{proof}
Let $y \in \rats_{\geq0}^k$ be the input to the $\gap$ problem. W.l.o.g., we can assume that $y \geq
\alpha \refpt$ for the reference point $\refpt$ computed by $A_1$, as otherwise there is no $y' \leq
\tfrac{1}{\alpha}y$ and $\gap$ can be answered negatively. 

We will solve the \gap\ problem with a single call of the $\beta$-approximation algorithm for
$\RPM(\calP, \Norm{\cdot}_p)$ (or $\RPM(\calP, \norm{\cdot}_p)$, respectively) with
\[p := \max \left\{\frac{\log k}{\log(1+\frac{1}{2M})}, 2kMq\right\} \;, \]
where $q$ is the largest denominator of all the components in $y$, and $M$ is an upper bound on the
objectives in $\calY$.
 
Let $I:=\{i\in[k]: y_i = \refpt_i = 0\}$. For $i \in [k]$, $\lambda_i = \begin{cases} 2 &\text{if }
i \in I, \\ \frac{1}{y_i-\refpt_i} &\text{otherwise.}
\end{cases}$

Let $y'$ be a $\beta$-approximation to $\min_{z \in \calY} r(z)$. 

If $r(y') \leq r(y)$, we return $y'$ as a positive answer to the $\gap$ problem. Observe that
$$\lambda_i(y'_i - \refpt_i) \leq \Norm{y' - \refpt}_p^\lambda \leq \Norm{y - \refpt}_p^\lambda \leq
1 + \tfrac{k}{p} \;.$$
If $i \in I$, we have $y'_i \leq \tfrac{1}{2}(1+\tfrac{1}{2Mq}) < 1$. If $i \notin I$, then $y'_i
\leq (1 + \tfrac{k}{p})y_i \leq y_i + \tfrac{1}{2q} < y_i + 1$. In both cases, integrality of $y'$
implies $y'_i \leq y_i$. The same holds for the $\norm{\cdot}_p$-norm with $1+\tfrac{k}{p}$ replaced
by $\sqrt[p]{k}$ -- in this case, the choice of $p$ guarantees $\sqrt[p]{k}\cdot y'_i < y'_i+1$.

%Multiplying the whole instance with $q$, we can apply \Cref{thm:cs-gives-all-pareto} to observe
%that the chosen weights ensure $y'_i \leq y_i$.

If $r(y') > r(y)$, we answer $\gap$ negatively: Let $y'' \in \calY$. We show that there is an $i \in
[k]$ with $ y''_i > \frac{1}{\alpha} y_i$. This is true if $y''_i > 0 = y_i$ for any $i \in I$.
Thus, we can restrict to the projection of $\rats^k$ to the components in $[k] \setminus I$, and
w.l.o.g. assume $I = \emptyset$. First observe that $\beta r(y'') \geq r(y') > r(y)$, which implies
$$\beta \Norm{y'' - \refpt}_p^\lambda > \Norm{y - \refpt}_p^\lambda - (\beta - 1)
\Norm{\refpt}_p^\lambda \;.$$
It is easy to verify that $\Norm{z}^\lambda_p \leq (1 + \tfrac{k}{p}) \norm{z}^\lambda_\infty$ for
all $z \in \rats^k$, and furthermore $\Norm{y - \refpt}^\lambda_p = 1 + \tfrac{k}{p}$, as $I =
\emptyset$. This yields
$$(1 + \tfrac{k}{p}) \beta \norm{y'' - \refpt}_\infty^\lambda > (1 + \tfrac{k}{p}) \norm{y -
\refpt}_\infty^\lambda - (1 + \tfrac{k}{p}) (\beta - 1) \norm{\refpt}_\infty^\lambda \;,$$
which brings us back to the case of the $\norm{\cdot}_\infty$-norm.
The same holds true for the $\norm{\cdot}_p$-norm, with the factor $1 + \tfrac{k}{p}$ replaced by
$\sqrt[p]{k}$.
%\qed
\end{proof}

% \spnewtheorem{thmrevisited4}[count-cor:cs-to-pareto]{Corollary}{\bfseries}{\itshape}
\begin{corollary}[Corollary~\ref{cor:cs-to-pareto} revisited, for $\Norm{\cdot}_p$ and
$\norm{\cdot}_p$] \label{cor:cs_p-to-pareto}
Let $\alpha > 1$ and set $\beta := \sqrt{\tfrac{\alpha^2}{2\alpha-1}}$. There is a polynomial time
algorithm for $\gap(\calP, \alpha)$, if there is a $\beta$-approximation algorithm for $\CS(\calP,
\Norm{\cdot}_p)$ ($\CS(\calP, \norm{\cdot}_p)$, respectively) for every $p \geq 1$ and the running
time of all algorithms is bounded by a polynomial in the instance size and $\log(p)$.
\end{corollary}

\begin{proof}
The proof is identical to that of Corollary~\ref{cor:cs-to-pareto} given in the paper. In fact, the
second
part of this proof only uses properties of monotone norms.%\qed
\end{proof}

\end{document}